\documentclass[journal]{IEEEtran}
\ifCLASSINFOpdf
\else
\fi
\usepackage{graphicx}
\usepackage{float}
\usepackage[cmex10]{amsmath}
\usepackage{amsmath,amssymb,amsthm}
\usepackage[compress]{cite}
\usepackage{balance}
\usepackage[UKenglish]{babel}
\usepackage{amssymb}
\usepackage{graphicx,picture,calc}
\usepackage{epstopdf}
\usepackage{subcaption}
\usepackage{mathtools}
\usepackage{booktabs}
\usepackage[compress]{cite}
\usepackage{enumerate}
\usepackage[norule]{footmisc}
\usepackage{algorithm}
\usepackage{algpseudocode}
\usepackage{amsmath}
\usepackage{array}
\usepackage{multirow}
\usepackage{dblfloatfix}
\usepackage{pifont}
\usepackage{color}
\usepackage{xcolor}

\makeatother

\newtheorem{Proposition}{\bf{Proposition}}
\newtheorem{Definition}{\bf{Definition}}
\newcolumntype{P}[1]{>{\centering\arraybackslash}p{#1}}
\newcolumntype{M}[1]{>{\centering\arraybackslash}m{#1}}

\begin{document}

\title{A General Conditional BER Expression of Rectangular QAM in the Presence of Phase Noise} 
\author{Thanh~V.~Pham,~
	    Thang~V.~Nguyen,~
        and~Anh~T.~Pham \\
\IEEEauthorblockA{$\textnormal{Computer Communications Lab., University of Aizu, Aizuwakamatsu, Japan}$\\
Emails: \{tvpham, d8212104, pham\}@u-aizu.ac.jp
}}
\maketitle 

\begin{abstract}
In this paper, we newly present a closed-form bit-error rate (BER) expression for an $M$-ary pulse-amplitude modulation ($M$-PAM) over  additive white Gaussian noise (AWGN) channels by analytically characterizing the bit decision regions and positions. The obtained expression is then used to derive the conditional BER of a rectangular quadrature amplitude modulation (QAM) for a given value of phase noise. Numerical results show that the impact of phase noise on the conditional BER performance is proportional to the constellation size. Moreover, it is observed that given a constellation size, the square QAM achieves the lowest phase noise-induced performance loss  compared to other rectangular constellations. 
\end{abstract}
\begin{IEEEkeywords}
QAM, Gray codes, bit-error rate, phase noise.
\end{IEEEkeywords}
\IEEEpeerreviewmaketitle
\section{Introduction}
Quadrature amplitude modulation (QAM) is a family of modulation schemes, which are widely used in today's communication systems due to their spectral efficiencies. Communication systems employing QAM requires a phase recovery for signal demodulation at the receiver. This phase recovery process is often not perfect due to hardware impairments and/or imperfect channel estimation \cite{Gappmair2017}. Consequently, phase noise can potentially be introduced to the received waveform causing an additional performance loss.

Regarding the impact of phase noise on the bit-error rate (BER) performance, there has been a number of studies for $M$-ary phase-shift keying (PSK) modulation ($M$-PSK) \cite{Jang2013,Song2015,Gappmair2017}, and references therein. Remarkably, for the case of rectangular QAM, an exact closed-form conditional BER expression under AWGN channels has been reported very recently \cite{Jafari2020}
where the derivation follows the similar approach presented in \cite{Cho2002} for the case without phase noise. Specifically, exact expressions were derived for the first two bits of symbols on the I channel. General expressions for arbitrary bits of the I and Q channels were then obtained based on the regularities of the particular Gray code sequence (labeling) being considered. It should be noted that there is a number of ways to label symbols using the Gray code in a QAM constellation, resulting in different BERs. As shown in \cite{Agrell2004} that under particular assumption on the channel, binary reflected Gray code (BRGC) sequences are the optimal labeling scheme for $M$-PSK, $M$-ary pulse-amplitude modulation ($M$-PAM), and rectangular QAM. Although BRGC sequences are equivalent in terms of the BER, each gives rise to a particular characterization of bit decision regions and positions, resulting in a different closed-form expression. Note that the authors in \cite{Jafari2020,Cho2002} assumed the original BRGC sequence proposed by Frank Gray in \cite{Gray1953} for their derivations.

In this paper, we aim at providing a rigorous derivation for the conditional BER of rectangular QAM over AWGN channels under the presence of phase noise. The obtained expression serves as a benchmark for evaluating the average BER in a future study.  In our previous work \cite{Nguyen2020}, we gave a sketch of the derivation and numerically calculated the BERs for two specific cases of $8 \times 4$ and $4 \times 4$ QAM. This work attempts at a complete analysis for the case of general rectangular QAM. Specifically, we first discuss in Section II the relationship between BRGC sequences and characterizations of bit decision regions in the case of $M$-PAM. This serves as the motivation for choosing the original BRGC sequence in our analysis. Unlike the previous works where closed-form expressions were obtained through generalization from specific cases \cite{Jafari2020,Cho2002}, our approach is to derive an exact closed-form expression for an $M$-PAM 
by analytically characterizing the bit decision regions and positions. The resulted expression can then be naturally applied to compute the conditional BER of a rectangular QAM conditioned on a given phase noise in Section III. Numerical and simulation results are given in Section IV to verify the accuracy of the analysis and to illustrate the impact of phase noise on the conditional BER.

\section{Bit Error Rate of PAM}
\subsection{Signal Model}
Let us consider an $M$-ary PAM, where $\log_2M$ is a positive integer. 
The signal waveforms can be expressed as 
\begin{align}
s(t) = A_M\cos2\pi f_c t, \hspace{5mm} 0 \leq t < T,
\end{align}
where $A_M$ is the signal amplitude, which is selected from the set $\{\pm d, \pm 3d, \hdots, \pm (M- 1)d\}$. $f_c$ is the carrier frequency, and $T$ is the symbol duration \cite{Cho2002}. The $M$ symbols are drawn from a BRGC sequence  with the distance between any two adjacent symbols being $2d$, which is given by
\begin{align}
d = \sqrt{\frac{3\log_2(M)E_b}{M^2-1}},
\end{align}
where $E_b$ is the bit energy. The received waveform after being distorted by an additive noise is then written as
\begin{align}
r(t) = s(t) + n(t),
\end{align}
where $n(t)$ denotes a zero-mean AWGN with two-sided power spectrum density $N_0/2$. In the following, the time index $t$ is omitted for simplicity. 
\begin{figure}[h]
\centering
\includegraphics[width = 8.8cm, height = 2.5cm]{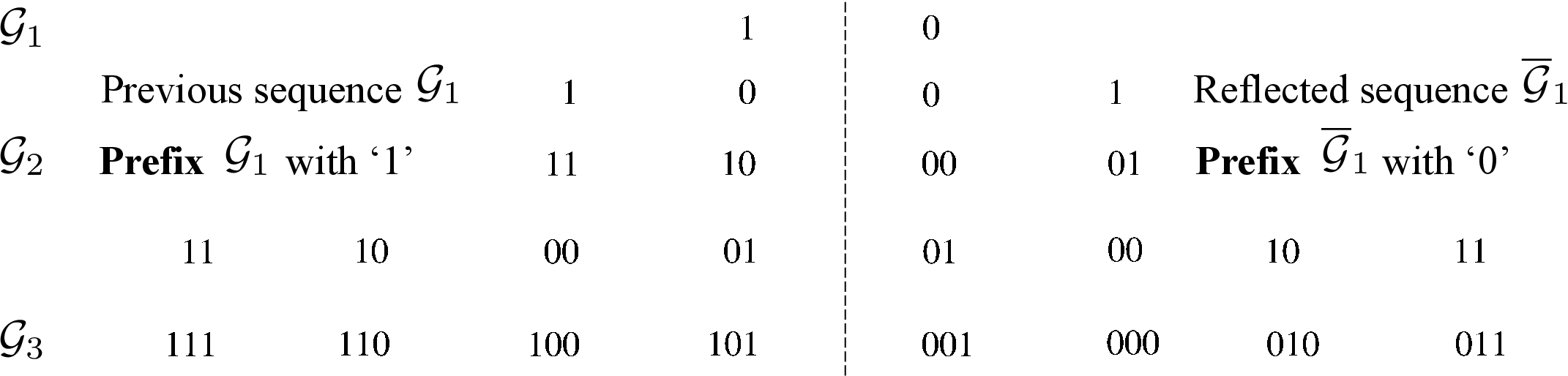}
\caption{Standard construction of $\mathcal{G}_3$.}
\label{GrayConstruction}
\end{figure}
\begin{figure*}[bt]
\centering
\begin{subfigure}[b]{0.48\textwidth}
\includegraphics[width = 8.2cm, height = 2.8cm]{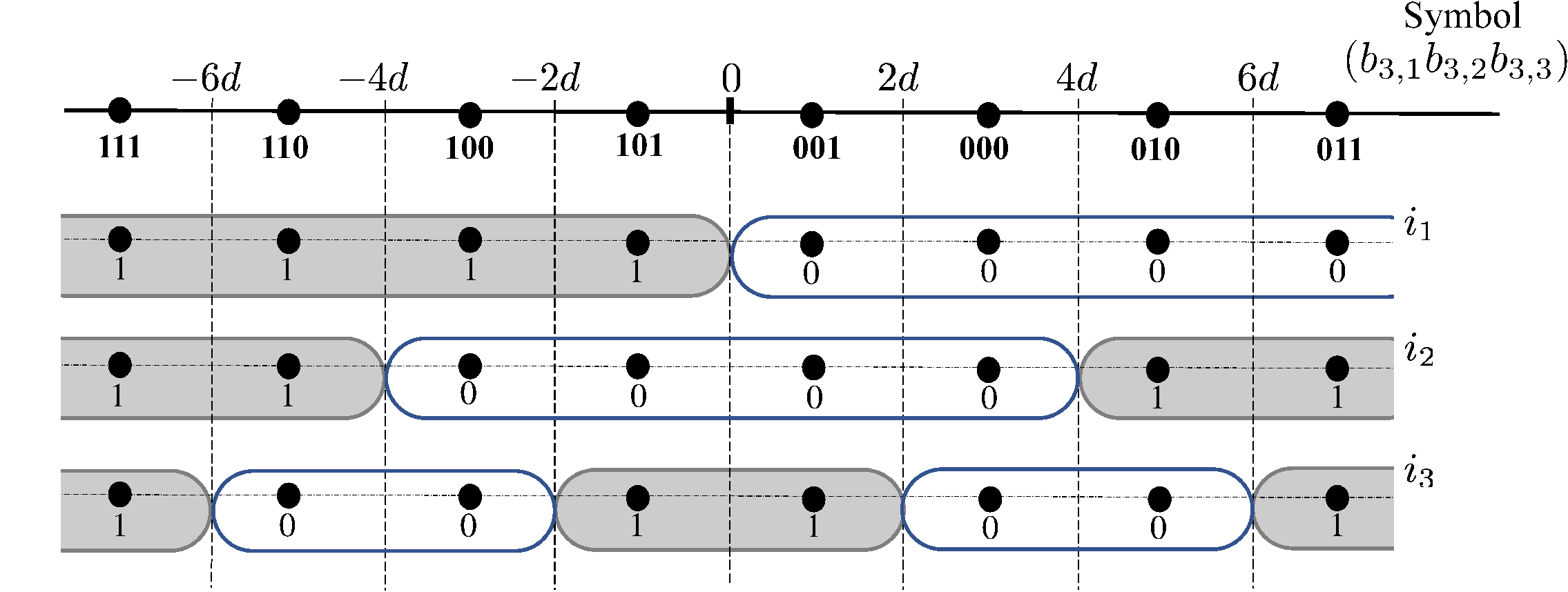}
\caption{Bit decision regions and positions according to $\mathcal{G}_3$.}
\label{bit_demaping_1}
\end{subfigure}
\centering
\begin{subfigure}[b]{0.48\textwidth}
\includegraphics[width = 8.2cm, height = 2.8cm]{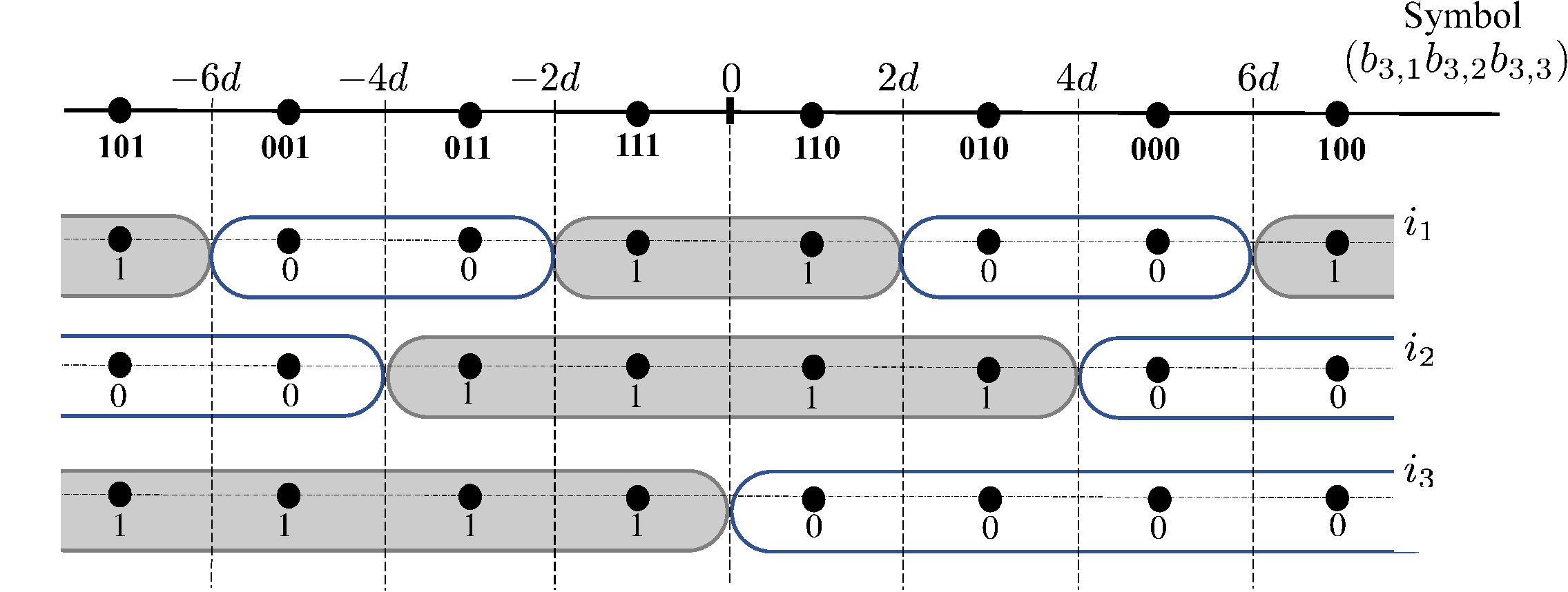}
\caption{Bit decision regions and positions according to $\mathcal{G}'_3$.}
\label{bit_demaping_2}
\end{subfigure}
\caption{An example of two different 8-PAM signal point arrangements.}
\label{decRegion-8PAM}
\end{figure*}
\subsection{Bit Decision Regions}
\subsubsection{Construction of Gray code sequences}
Since our analytical framework directly depends on Gray code sequences' construction, specifically the BRGC, we discuss this issue in this section. For the sake of readability, let us revisit the construction of BRGC sequences proposed in \cite{Gray1953}. According to this approach, let $\mathcal{G}_1 = (1, 0)$ be the trivial 1-bit sequence. For $n \geq 2$, denote $\mathcal{G}_{n-1}$ as an $(n-1)$-bit BRGC sequence and $\overline{\mathcal{G}}_{n-1}$ as the reflected version of $\mathcal{G}_{n-1}$ (i.e., codewords of $\mathcal{G}_{n-1}$ are placed in the reverse order). An extra bit `1' and `0' are added to each codeword of $\mathcal{G}_{n-1}$ and $\overline{\mathcal{G}}_{n-1}$ from the left (i.e., prefixing), respectively. An $n$-bit BRGC sequence $\mathcal{G}_n$ is then obtained by concatenating the two newly formed sequences. We regard this as the standard construction, and the resulted sequence is called the standard BRGC sequence. For example, Fig. \ref{GrayConstruction} illustrates the generation of the standard 3-bit BRGC sequence $\mathcal{G}_3$. Note that, from the standard $\mathcal{G}_n$ sequence, it is possible to generate another $n!2^n - 1$ sequences from it by means of permutation and/or complementation of all the codewords' bits in the same manner \cite{Arazi1984}. For example, permuting the first and third bits and complementing the second bit in all codewords of $\mathcal{G}_3$ result in a new Gray code sequence $\mathcal{G}'_3 = (101~001~011~111~110~010~000~100)$.

In this paper, we are interested in an $M$-PAM where the arrangement of symbols\footnote{We use the terms ``symbols" and ``codewords" interchangeably.} is drawn from the standard $\mathcal{G}_{\log_2(M)}$ BRGC sequence. By convention, codewords in the first and second half of the sequence are placed in the left (negative $A_M$ values) and right (positive $A_M$ values) half plane of the signal space, respectively. Denote $K = \log_2(M)$ and let $b_{K,k}$ be the symbol's $k-$th bit from the left ($k \leq K$). It is worth mentioning that different code sequences lead to different decision regions of $b_{K, k}$ as illustrated in Fig. \ref{decRegion-8PAM} for the case of 8-PAM with symbol arrangements according to $\mathcal{G}_3$ and $\mathcal{G}'_3$. In this example, although the two sequences result in different decision regions of each bit (hence its error rate), the overall BERs of the two cases are the same.  In fact, \cite{Agrell2004} shows that trivial operations on codewords of a sequence such as bit permutation and complementation do not affect its BER. Hence, it suffices to derive the BER of an $M$-PAM whose symbol arrangement is drawn from the standard $\mathcal{G}_{K}$ sequence. 
\subsubsection{Analytical characterization of bit decision regions}
To facilitate the discussion, the following definition and proposition are necessary. 
\begin{Definition}
Let $b^i_{K, 1}b^i_{K,2}...b^i_{K,K}$ ($b^i_{K, k} \in \{0, 1\}$) be the $i-$th codeword of $\mathcal{G}_K$. Then, $\mathcal{B}_{K, k} = \left(b^1_{K, k}, b^2_{K, k}, \cdots, b^{2^K}_{K, k}\right)$ is defined as a binary $2^K$-tuple  whose $i-$th element  is the value of  $b_{K, k}$ in the $i-$th codeword of $\mathcal{G}_K$. 
\end{Definition} 
\begin{Proposition}
 \begin{align}
\mathcal{B}_{K, k} \!=\! \left\{ {\begin{array}{*{20}{l}}
\!\!\!(\underbrace{1, \cdots, 1}_{2^{K-1}}, \underbrace{0, \cdots, 0}_{2^{K-1}}) \hspace*{0pt}\hfill {{\text{if~~}}} k = 1, \\
\!\!\!(\underbrace{1, \cdots, 1}_{2^{K-k}}, \underbrace{0, \cdots, 0}_{2^{K-k + 1}}, \underbrace{1, \cdots, 1}_{2^{K-k + 1}}, \cdots, \underbrace{0, \cdots, 0}_{2^{K-k + 1}},  \underbrace{1, \cdots, 1}_{2^{K-k }}) \nonumber \\ \hspace*{0pt}\hfill {{\text{if~~}}}  k \geq 2. \\
\end{array}} \right.
\end{align}
\end{Proposition}

\begin{proof} Due to the construction rule of $\mathcal{G}_K$ from $\mathcal{G}_{K-1}$ which has $2^{K-1}$ codewords, it is obvious to see that $\mathcal{B}_{K, 1} = (\underbrace{1, \cdots, 1}_{2^{K-1}}, \underbrace{0, \cdots, 0}_{2^{K-1}})$. For $k \geq 2$, let $\overline{\mathcal{B}}_{K, k}$ be the reflected version of $\mathcal{B}_{K, k}$. Notice that $\mathcal{B}_{K, k}$ is symmetrical since $\mathcal{B}_{K, k}  = \left(\mathcal{B}_{K-1, k-1}\overline{\mathcal{B}}_{K-1, k-1}\right)$. Thus, $\mathcal{B}_{K, k} = \overline{\mathcal{B}}_{K, k}$ and we can write 
\begin{align}
\mathcal{B}_{K, k} & = \left(\mathcal{B}_{K-1, k-1}\overline{\mathcal{B}}_{K-1, k-1}\right) \nonumber \\
			     & = \left(\mathcal{B}_{K-2, k-2}\overline{\mathcal{B}}_{K-2, k-2}{\mathcal{B}}_{K-2, k-2}\overline{\mathcal{B}}_{K-2, k-2}\right) \nonumber \\
			     & ~~\vdots \nonumber \\
			     & = (\underbrace{\mathcal{B}_{K-k+1, 1}\overline{\mathcal{B}}_{K-k+1, 1}\cdots\mathcal{B}_{K-k+1,1}\overline{\mathcal{B}}_{K-k+1, 1}}_{2^{k-1} \text{pairs}~\mathcal{B}_{K-k+1,1}\overline{\mathcal{B}}_{K-k+1, 1}} ).
\label{bit-pattern}
\end{align}
Since $\mathcal{B}_{K-k+1, 1} = (\underbrace{1, \cdots, 1}_{2^{K-k}}, \underbrace{0, \cdots, 0}_{2^{K-k}})$ and $\overline{\mathcal{B}}_{K-k+1, 1} = (\underbrace{0, \cdots, 0}_{2^{K-k}}, \underbrace{1, \cdots, 1}_{2^{K-k}})$, the preposition is proved following the observation from \eqref{bit-pattern}. \qedhere
\end{proof}

The proposition shows in the trivial case of $k = 1$ that $b_{K, 1} = 1$ if $r \leq 0$ and $b_{K, 1} = 0$ if $r > 0$. For $k \geq 2$, a decision rule for $b_{K, k}$ can be written as 
\begin{align}
b_{K, k} = \left\{ {\begin{array}{*{20}{c}}
1 & {{\text{if~~}}}& r \odot \mathcal{R}^{(1)}_{K, k}, \\
0 & {{\text{if~~}}}& r \odot \mathcal{R}^{(0)}_{K, k}, \\
\end{array}} \right.
\end{align}
where $\mathcal{R}^{(1)}_{K, k} = \left\{\left(\delta^{(1), l}_{K, k}d,~\delta^{(1), u}_{K, k}d\right)~\Big|~\delta^{(1), l}_{K, k} \in \left\{\mathbb{Z}, -\infty\right\}, \right.$ $\left. \delta^{(1), u}_{K, k} \in \left\{\mathbb{Z}, +\infty\right\}\right\}$ and $\mathcal{R}^{(0)}_{K, k} = \left\{\left(\delta^{(0), l}_{K, k}d,~\delta^{(0), u}_{K, k}d\right)~\Big|~\delta^{(0), l}_{K k} \in \mathbb{Z},  \delta^{(0), u}_{K, k} \in \mathbb{Z}\right\}$ are the sets of decision regions of bits `1' and `0', respectively. Note that, bit `1's decision regions at the two extremes of the signal space are unbounded. This is represented by $ \delta^{(1), l}_{K, k} = -\infty$ and $ \delta^{(1), u}_{K, k} = +\infty$. With $k = 3$ from the above example, $\mathcal{R}^{(1)}_{3, 3} = \left\{\left(-\infty, -6d\right), ~\left(-2d, 2d\right),~\left(6d, +\infty\right)\right\}$ and $\mathcal{R}^{(0)}_{3, 3} = \left\{\left(-6d, -2d\right), ~\left(2d, 6d\right)\right\}$. We use the notation $r  \odot \mathcal{R}^{(1)}_{K, k}$ (or  $r  \odot \mathcal{R}^{(0)}_{K, k}$) to imply that there exist some $\left(\delta^{(1), l}_{K, k}d,~\delta^{(1), u}_{K, k}d\right) \in \mathcal{R}^{(1)}_{K, k}$ that satisfies $\delta^{(1), l}_{K, k}d \leq r <  \delta^{(1), u}_{K, k}d$. We now derive analytical expressions of $\left\{\delta^{(0), l}_{K, k}\right\}$,  $\left\{\delta^{(0), u}_{K, k}\right\}$, $\left\{\delta^{(1), l}_{K, k}\right\}$, and $\left\{\delta^{(1), u}_{K, k}\right\}$. 

A direct implication of the proposition is that the number of `1's and `0's regions of $b_{K, k}$ is $\left|\mathcal{R}^{(1)}_{K, k}\right| + \left|\mathcal{R}^{(0)}_{K, k}\right| = 2^{k-1}+ 1$.
Accordingly, let $\mathbf{p}_k = \begin{bmatrix}-1 & 0 & 1 & \cdots & 2^{k-1}-1 & 2^{k-1} \end{bmatrix}$ be the vector whose two consecutive elements represent the starting and ending points of a region. 
Denote $\Delta(x, K, k)$ as a function of $x$, $K$, and $k$, which returns the values of $\left\{\delta^{(0), l}_{K, k}\right\}$,  $\left\{\delta^{(0), u}_{K, k}\right\}$, $\left\{\delta^{(1), l}_{K, k}\right\}$, and $\left\{\delta^{(1), u}_{K, k}\right\}$ for $x$ taking on values from $\mathbf{p}_k$. 
To account for $ \delta^{(1), l}_{K, k} = -\infty$ and $ \delta^{(1), u}_{K, k} = +\infty$, it is  mathematically convenient to use the convention $\Delta(-1, K, k) = -\infty$ and $\Delta(2^{k-1}, K, k) = +\infty$. 
Observe that there are $2^{K-k}$ bits `1' in the first decision regions followed by interleaving regions of `0' and `1' each containing $2^{K-k + 1}$ bits  and ended by a region (`0' if $k = 1$ or `1' if $k \geq 2$) having $2^{K-k}$ bit. As the distance between any two adjacent codewords is $2d$ with the first one starting at $(-2^K+1)d$, the observation gives rise to the definition of $\Delta(x, K, k)$ as
\begin{align}
&\Delta(x, K, k) = \nonumber \\
& \left\{ {\begin{array}{*{5}{l}}
\!\!\!-\infty & , x = -1, \\
\!\!\!-2^K\left(1 - \frac{1}{2^{k-1}}\right) + x\times2^{K - k +2} &, -1 < x < 2^{k-1} , \\
\!\!\!+\infty &, x =  2^{k-1}.
\end{array}} \right.
\label{delta}
\end{align}
Then, $\left\{\delta^{(1), l}_{K, k}\right\}$, $\left\{\delta^{(1), u}_{K, k}\right\}$, $\left\{\delta^{(0), l}_{K, k}\right\}$, and  $\left\{\delta^{(0), u}_{K, k}\right\}$ are given by 
\begin{align}
&\left\{\delta^{(1), l}_{K, k}\right\} = \Delta\left(2\mathbf{p}_k^{(1)}+1, K, k\right), 
\label{delta_1_l}\\
&\left\{\delta^{(1), u}_{K, k}\right\} = \Delta\left(2\mathbf{p}_k^{(1)}+2, K, k\right),
\label{delta_1_u}\\
&\left\{\delta^{(0), l}_{K, k}\right\} = \Delta\left(2\mathbf{p}_k^{(0)} + 2, K, k\right), 
\label{delta_0_l}\\
&\left\{\delta^{(0), u}_{K, k}\right\} = \Delta\left(2\mathbf{p}_k^{(0)}+3, K, k\right), 
\label{delta_0_u}
\end{align}
where 
\begin{align}
\mathbf{p}_k^{(1)} = \left\{ {\begin{array}{*{20}{l}}
-1 & {{\text{if~~}}}  & k = 1,\\
\begin{bmatrix}-1 & 0 & \cdots  & 2^{k-2} - 1\end{bmatrix} & {{\text{if~~}}} & k \geq 2,
\end{array}} \right.
\label{region-1-index}
\end{align}
and
\begin{align}
\mathbf{p}_k^{(0)} = \left\{ {\begin{array}{*{20}{l}}
-1 & {{\text{if~~}}}  & k = 1,\\
\begin{bmatrix}-1 & 0 & \cdots & 2^{k-2} - 2 \end{bmatrix}& {{\text{if~~}}} & k \geq 2.
\end{array}} \right.
\label{region-0-index}
\end{align}
\subsection{Bit Positions and Conditional BER Expression}
\begin{figure*}[ht]
\begin{align}
&P_{\text{PAM}}(b_{K, k} )   = \text{Pr}\left(r \odot \mathcal{R}^{(1)}_{K, k} \Big| s \in \mathcal{A}^{(0)}_{K, k} \right) 
               +  \text{Pr}\left(r \odot \mathcal{R}^{(0)}_{K, k} \Big| s \in \mathcal{A}^{(1)}_{K, k} \right)
               = \frac{1}{{MK}}\left(\sum_{s \in \mathcal{A}^{(0)}_{K, k} }\!\text{Pr}\left(r \odot \mathcal{R}^{(1)}_{K, k} \right) 
               \! + \!\sum_{s \in \mathcal{A}^{(1)}_{K, k} }\!\text{Pr}\left(r \odot \mathcal{R}^{(0)}_{K, k} \right) \right) \nonumber \\
              & = \frac{1}{MK}\left(\sum_{\overline{p}^{(0)} \in \mathbf{\overline{p}}_k^{(0)}}\sum_{j^{(0)} \in \mathbf{j}^{\overline{p}^{(0)}}_{K, k}}\sum_{p^{(1)} \in \mathbf{p}_k^{(1)}}\text{Pr}\left(\delta^{p^{(1)}, l}_{K, k}d \leq A^{\overline{p}^{(0)}, j^{(0)}}_{K, k}d + n < \delta^{p^{(1)}, u}_{K, k}d\right) \right. \nonumber \\ 
              &~~~~~~~~~ \left.+ \sum_{\overline{p}^{(1)} \in \mathbf{\overline{p}}_k^{(1)}}\sum_{j^{(1)} \in \mathbf{j}^{\overline{p}^{(1)}}_{K, k}}\sum_{p^{(0)} \in \mathbf{p}_k^{(0)}}\text{Pr}\left(\delta^{p^{(0)}, l}_{K, k}d \leq A^{\overline{p}^{(1)}, j^{(1)}}_{K, k}d + n < \delta^{p^{(0)}, u}_{K, k}d\right) \right)  \nonumber \\
              & = \frac{1}{2MK}\left(\sum_{\overline{p}^{(0)} \in \mathbf{\overline{p}}_k^{(0)}}\sum_{j^{(0)} \in \mathbf{j}^{\overline{p}^{(0)}}_{K, k}}\sum_{p^{(1)} \in \mathbf{p}^{(1)}_k}\text{erfc}\left(\frac{d}{\sqrt{N_0}}\left(\delta^{p^{(1)}, l}_{K, k} - A^{\overline{p}^{(0)}, j^{(0)}}_{K, k}\right)\right) - \text{erfc}\left(\frac{d}{\sqrt{N_0}}\left(\delta^{p^{(1)},u}_{K, k} - A^{\overline{p}^{(0)}, j^{(0)}}_{K, k}\right)\right) \right . \nonumber \\ & \left. ~~~~~~~~~~~ +\sum_{\overline{p}^{(1)} \in \mathbf{\overline{p}}_k^{(1)}}\sum_{j^{(1)} \in \mathbf{j}^{\overline{p}^{(1)}}_{K, k}}\sum_{p^{(0)} \in \mathbf{p}^{(0)}_k}  \text{erfc}\left(\frac{d}{\sqrt{N_0}}\left(\delta^{p^{(0)}, l}_{K, k} - A^{\overline{p}^{(1)}, j^{(1)}}_{K, k} \right)\right)  - \text{erfc}\left(\frac{d}{\sqrt{N_0}}\left(\delta^{p^{(0)}, u}_{K, k} - A^{\overline{p}^{(1)}, j^{(1)}}_{K, k}\right)\right)  \vphantom{\sum_{\left(\delta^{(1), l}_{K, k}d, \delta^{(1), u}_{K, k}d\right) \in \mathcal{R}^{(1)}_{K, k} }\sum_{q = -\frac{M}{2} }^{\frac{M}{2} - 1}} \right).
\label{BER-PAM}             
\end{align}
\noindent\makebox[\linewidth]{\rule{\textwidth}{1pt}}
\end{figure*}
To formulate an expression for the BER of $b_{K, k}$, it is now necessary to determine the positions of the elements of $\mathcal{B}_{K, k}$ in the signal space. 
In order to so, let $\mathcal{A}^{(1)}_{K, k} = \left\{A^{(1)}_{K, k}d ~ \big| ~ A^{(1)}_{K, k} \in \mathbb{Z} \right\}$ and $\mathcal{A}^{(0)}_{K, k} = \left\{A^{(0)}_{K, k}d ~ \big| ~ A^{(0)}_{K, k} \in \mathbb{Z} \right\}$ be the sets of positions of the symbols whose $b_{K, k}$'s are `1' and `0', respectively. In our example of $8$-PAM in Fig. \ref{bit_demaping_1}, for $k = 3$, we have $\mathcal{A}^{(1)}_{3, 3} = \left\{-7d, -d, d, 7d  \right\}$ and $\mathcal{A}^{(0)}_{3, 3} = \left\{-5d, -3d, 3d, 5d  \right\}$. Since the values of $\mathcal{A}^{(1)}_{K, k}$ $(\mathcal{A}^{(0)}_{K, k})$ lie on $\mathcal{R}^{(1)}_{K, k}$ $(\mathcal{R}^{(0)}_{K, k})$, they can be determined using the defined $\Delta(\cdot)$ function. To make characterizations of  $\mathcal{A}^{(0)}_{K, k}$ and $\mathcal{A}^{(1)}_{K, k}$ convenient, 
we bound the two unbounded regions at $-2^K$ and $2^K$, which correspond to $\Delta\left(-\frac{1}{2}, K, k\right)$ and $\Delta\left(2^{k-1} - \frac{1}{2}, K, k\right)$, respectively. This naturally necessitates a definition of $\mathbf{\overline{p}}_k = \begin{bmatrix}-\frac{1}{2} & 0 & 1 &\cdots & 2^{k-1} - 1& 2^{k-1} - \frac{1}{2} \end{bmatrix}$. 
Hence, we can write $\mathcal{A}^{(1)}_{K, k} = \left\{\left\{A_{K, k}^{(1)}\right\}_{\left[\mathbf{\overline{p}}^{(1)}_k\right]_i}\right\}$ (similarly for $\mathcal{A}^{(0)}_{K, k}$), where $\left\{A_{K, k}^{(1)}\right\}_{\left[\mathbf{\overline{p}}^{(1)}_k\right]_i}$ is the set of bit `1's positions corresponding to $\left[\mathbf{\overline{p}}^{(1)}_k\right]_i$, which is the $i-$th element of $\mathbf{\overline{p}}^{(1)}_k$ and is defined as
\begin{align}
\mathbf{\overline{p}}_k^{(1)} = \left\{ {\begin{array}{*{20}{l}}
-\frac{3}{4} & {{\text{if~~}}}  & k = 1,\\
\begin{bmatrix}-\frac{3}{4} & 0 & \cdots  & 2^{k-2} - 1\end{bmatrix} & {{\text{if~~}}} & k \geq 2.
\end{array}} \right.
\label{region-1-index}
\end{align}
Similarly, denote $\left[\mathbf{\overline{p}}^{(0)}_k\right]_i$ as the $i-$th element of $\mathbf{\overline{p}}^{(0)}_k$, which is defined to be the same as  $\mathbf{{p}}^{(0)}_k$ in \eqref{region-0-index}.
We also define the length of a region $\left(\Delta\left([\mathbf{\overline{p}}_k]_i, K, k\right)~\Delta\left([\mathbf{\overline{p}}_k]_{i+1}, K, k\right)\right)$ straightforwardly as
\begin{align}
l\left([\mathbf{\overline{p}}_k]_i, K, k\right) &= \Delta([\mathbf{\overline{p}}_k]_{i + 1}, K, k) - \Delta\left([\mathbf{\overline{p}}_k]_i, K, k\right)\nonumber \\
& = \left([\mathbf{\overline{p}}_k]_{i+1} - [\mathbf{\overline{p}}_k]_i\right)\times2^{K - k +2}.
\label{delta-modified}
\end{align}
Now, $\left\{A_{K, k}^{(1)}\right\}_{\left[\mathbf{\overline{p}}^{(1)}_k\right]_i} $ and $\left\{A_{K, k}^{(0)}\right\}_{\left[\mathbf{\overline{p}}^{(0)}_k\right]_i}$ are respectively given by
\begin{align}
\left\{A_{K, k}^{(1)}\right\}_{\left[\mathbf{\overline{p}}^{(1)}_k\right]_i}  = \Delta\left(2\left[\mathbf{\overline{p}}^{(1)}_k\right]_i + 1, K, k\right) + 2\mathbf{j}_{K, k}^{\left[\mathbf{\overline{p}}^{(1)}_k\right]_i} + 1,
\label{position_1}
\end{align}
and
\begin{align}
\left\{A_{K, k}^{(0)}\right\}_{\left[\mathbf{\overline{p}}^{(0)}_k\right]_i} = \Delta\left(2\left[\mathbf{\overline{p}}^{(0)}_k\right]_i + 2, K, k\right) + 2\mathbf{j}_{K, k}^{\left[\mathbf{\overline{p}}^{(0)}_k\right]_i} + 1,
\label{position_0}
\end{align}
, where $\mathbf{j}^{\left[\mathbf{\overline{p}}^{(1)}_k\right]_i}_{K, k} =  \begin{bmatrix} 0 & 1 & \cdots & \frac{l\left(2\left[\mathbf{\overline{p}}^{(1)}_k\right]_i + 1, K, k\right)}{2} -1 \end{bmatrix}$ and $\mathbf{j}^{\left[\mathbf{\overline{p}}^{(0)}_k\right]_i}_{K, k} = \begin{bmatrix} 0 & 1 & \cdots & \frac{l\left(2\left[\mathbf{\overline{p}}^{(0)}_k\right]_i + 2, K, k\right)}{2} -1 \end{bmatrix}$.

For each transmitted symbol $s$, $b_{K,k}$ is decoded wrongly when $r$ falls into bit `1's decision regions given that $s$ is chosen from the sets of `0's positions or vice-versa. For expressional convenience, let $\delta^{p^{(1)}, l}_{K, k}$,  $\delta^{p^{(1)}, u}_{K, k}$, $\delta^{p^{(0)}, l}_{K, k}$, and $\delta^{p^{(0)}, u}_{K, k}$ be the elements of $\left\{\delta^{(1), l}_{K, k}\right\}$, $\left\{\delta^{(1), u}_{K, k}\right\}$, $\left\{\delta^{(0), l}_{K, k}\right\}$,  and $\left\{\delta^{(0), u}_{K, k}\right\}$ in \eqref{delta_1_l}-\eqref{delta_0_u}, which correspond to some $p^{(1)}_k \in \mathbf{p}^{(1)}_k$ and $p^{(0)}_k \in \mathbf{p}^{(0)}_k$.
Also, let $A^{\overline{p}^{(1)}, j^{(1)}}_{K, k}$ and $A^{\overline{p}^{(0)}, j^{(0)}}_{K, k}$ are the elements of $\left\{A_{K, k}^{(1)}\right\}_{\left[\mathbf{\overline{p}}^{(1)}_k\right]_i} $ and $\left\{A_{K, k}^{(0)}\right\}_{\left[\mathbf{\overline{p}}^{(0)}_k\right]_i}$, which correspond to some $\overline{p}^{(1)}_{K, k} \in \overline{\mathbf{p}}^{(1)}_{K, k}$, $\overline{p}^{(0)}_{K, k} \in \overline{\mathbf{p}}^{(0)}_{K, k}$, $j^{(1)} \in \mathbf{j}^{\overline{p}^{(1)}_{K, k}}_{K, k}$, and  $j^{(0)} \in \mathbf{j}^{\overline{p}^{(0)}_{K, k}}_{K, k}$. The error probability of $b_{K, k}$ can then be expressed in \eqref{BER-PAM}, which is on top of the previous page. Note that, $\text{erfc}(x) = \frac{2}{\sqrt{\pi}}\int_{x}^{\infty}e^{-t^2}dt$ is the complementary error function. 

\section{Bit Error Rate of QAM with Phase Noise}
We now consider a rectangular $M_I \times M_Q$ QAM, whose signal constellation can be seen as a combination of two PAM signalings. The signal waveform is expressed by
\begin{align}
s(t) = A_{I}\cos\left(2\pi f_c t\right) - A_{Q}\sin\left(2\pi f_c t\right),
\end{align}
where  $A_{I}$ and $A_{Q}$ are the signal amplitudes with respect to the I and Q channels, respectively. Also, let $d_I$ and $d_Q$ be the half distances between two adjacent symbols on the two channels. Without lost of generality, we assume that $d_I = d_Q = d$, which is given by
\begin{align}
d = \sqrt{\frac{3\log_2\left(M_IM_Q\right)E_b}{M_I^2 + M_Q^2 - 2}}.
\end{align}
The received waveform after being distorted by an AWGN and phase noise can be written as
\begin{align}
r(t) \!=\! A_{I}\cos\left(2\pi f_c t + \theta_e \right) \!- \! A_{Q}\sin\left(2\pi f_c t + \theta_e \right) + n,
\end{align}
where the phase noise component $\theta_e \in [-\pi ~ \pi]$ caused by the imperfections of the phase-locked loop and/or channel estimation can be modeled by a Tikhonov distribution \cite{Gappmair2017}. The I and Q components of the received waveform are then given by
\begin{align}
r_I&= A_I\cos(\theta_e) - A_Q\sin(\theta_e) + n,   \\
r_Q &= A_I\sin(\theta_e) + A_Q\cos(\theta_e) + n.
\end{align}
Denote $m_I = \log_2(M_I)$ and $m_Q = \log_2(M_Q)$, the conditional BER of the $M_I \times M_Q$ QAM can be calculated  by summing up the conditional BERs of the corresponding $m_I$- and $m_Q$-PAM, resulting in
\begin{align}
P_{\text{QAM}}(\theta_e) = \sum_{k = 1}^{m_I} P_{\text{e}}(i_{m_I, k}|\theta_e) + \sum_{k = 1}^{m_Q} P_{\text{e}}(q_{m_Q,k}|\theta_e),
\end{align}
where $i_{m_I, k}$ and $q_{m_Q, k}$ are the $k-$th bits of the symbols of the $M_I$- and $M_Q$-PAM. 
Following the BER expression derived for the $M$-PAM in \eqref{BER-PAM}, the conditional BERs of  $i_{m_I, k}$ and $q_{m_Q, k}$ conditioned on $\theta_e$ denoted as $P_{\text{e}}(i_{m_I, k}|\theta_e)$ and $P_{\text{e}}(q_{m_Q, k}|\theta_e)$ are given by 
\begin{align}
&P_{\text{e}}(i_{m_I, k} | \theta_e) = \frac{1}{2(m_I + m_Q)M_I M_Q} \nonumber \\
& \times \left(\sum_{\overline{p}^{(0)} \in \mathbf{\overline{p}}_k^{(0)}}\sum_{j^{(0)} \in \mathbf{j}^{\overline{p}^{(0)}}_{m_I, k}}\sum_{p^{(1)} \in \mathbf{p}^{(1)}_k}\sum_{q = -\frac{M_Q}{2}}^{\frac{M_Q}{2}-1}\right. \nonumber \\ & \left.  ~~~~~~~\Psi\left(\frac{d}{\sqrt{N_0}}, \delta^{p^{(1)}, l}_{m_I, k}, -A^{\overline{p}^{(0)}, j^{(0)}}_{m_I, k}, 2q+1, \theta_e\right) \right. \nonumber \\ & \left. ~~~~~- \Psi\left(\frac{d}{\sqrt{N_0}},\delta^{p^{(1)},u}_{m_I, k}, -A^{\overline{p}^{(0)}, j^{(0)}}_{m_I, k}, 2q+1, \theta_e\right)  \right. \nonumber \\ & \left. ~~+\sum_{\overline{p}^{(1)} \in \mathbf{\overline{p}}_k^{(1)}}\sum_{j^{(1)} \in \mathbf{j}^{\overline{p}^{(1)}}_{m_I, k}}\sum_{p^{(0)} \in \mathbf{p}^{(0)}_k}\sum_{q = -\frac{M_Q}{2}}^{\frac{M_Q}{2}-1} \right. \nonumber \\ & \left.  ~~~~~~~\Psi\left(\frac{d}{\sqrt{N_0}}, \delta^{p^{(0)}, l}_{m_I, k},  -A^{\overline{p}^{(1)}, j^{(1)}}_{m_I, k}, 2q+1, \theta_e \right)  \right. \nonumber \\ & \left. ~~~~~- \Psi\left(\frac{d}{\sqrt{N_0}}, \delta^{p^{(0)}, u}_{m_I, k}, -A^{\overline{p}^{(1)}, j^{(1)}}_{m_I, k}, 2q+1, \theta_e\right)  \vphantom{\sum_{\left(\delta^{(1), l}_{K, k}d, \delta^{(1), u}_{m_I, k}d\right) \in \mathcal{R}^{(1)}_{K, k} }\sum_{q = -\frac{M}{2} }^{\frac{M}{2} - 1}} \!\!\! \right),
\label{BER-QAM-I}             
\end{align}
and 
\begin{align}
&P_{\text{e}}(q_{m_Q, k} |\theta_e)  = \frac{1}{2(m_I + m_Q)M_IM_Q} \nonumber \\ & \times \left(\sum_{\overline{p}^{(0)} \in \mathbf{\overline{p}}_k^{(0)}}\sum_{j^{(0)} \in \mathbf{j}^{\overline{p}^{(0)}}_{m_Q, k}}\sum_{p^{(1)} \in \mathbf{p}^{(1)}_k}\sum_{i = -\frac{M_I}{2}}^{\frac{M_I}{2}-1} \right. \nonumber \\ & \left. ~~~~~~~\Psi\left(\frac{d}{\sqrt{N_0}}, \delta^{p^{(1)}, l}_{m_Q, k}, -A^{\overline{p}^{(0)}, j^{(0)}}_{m_Q, k}, - (2i+1), \theta_e\right)  \right. \nonumber \\ & \left. ~~~~~- \Psi\left(\frac{d}{\sqrt{N_0}}, \delta^{p^{(1)},u}_{m_Q, k}, -A^{\overline{p}^{(0)}, j^{(0)}}_{m_Q, k}, - (2i +1), \theta_e\right) \right . \nonumber \\ & \left.~~  +\sum_{\overline{p}^{(1)} \in \mathbf{\overline{p}}_k^{(1)}}\sum_{j^{(1)} \in \mathbf{j}^{\overline{p}^{(1)}}_{m_Q, k}}\sum_{p^{(0)} \in \mathbf{p}^{(0)}_k}\sum_{i = -\frac{M_I}{2}}^{\frac{M_I}{2}-1} \right. \nonumber \\ & \left. ~~~~~~~\Psi\left(\frac{d}{\sqrt{N_0}}, \delta^{p^{(0)}, l}_{m_Q, k}, -A^{\overline{p}^{(1)}, j^{(1)}}_{m_Q, k}, - (2i+1), \theta_e\right) \right. \nonumber \\ & \left.  ~~~~~- \Psi\left(\frac{d}{\sqrt{N_0}}, \delta^{p^{(0)}, u}_{m_Q, k}, -A^{\overline{p}^{(1)}, j^{(1)}}_{m_Q, k}, -(2i+1), \theta_e\right)  \vphantom{\sum_{\left(\delta^{(1), l}_{K, k}d, \delta^{(1), u}_{m_Q, k}d\right) \in \mathcal{R}^{(1)}_{K, k} }\sum_{i = -\frac{M}{2} }^{\frac{M}{2} - 1}} \!\!\! \right),
\label{BER-QAM-Q}             
\end{align}
respectively. In the above expressions, we define $\Psi(\alpha, x, y, z, \varphi) = \text{erfc}\left(\alpha\left(x+y\cos(\varphi)+z\sin(\varphi)\right)\right)$. 

\section{Numerical Examples and Dicussions}
Firstly, we show in Fig. \ref{square_QAM} the conditional BER of square $M$-QAM with different constellation sizes as a function of ${E_b}/{N_0}$. In case of the performance with phase noise, $\theta_e = \pi/180$ is chosen. Similar to the case of PSK, we observed significant performance losses due to phase noise when the constellation size is large. At conditional BER = $10^{-3}$, performance losses are negligible in the case of $4$- and $16$-QAM. However, these are about $0.25$ dB, $1$ dB, and $3.6$ dB in the case of $64$-, $256$-, and $1024$-QAM, respectively.
\begin{figure}[ht]
\centering
\includegraphics[scale = 0.315]{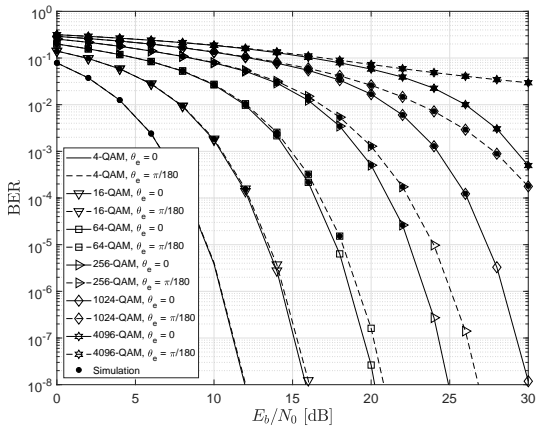}
\caption{BER of square QAM for different constellation sizes. }
\label{square_QAM}
\end{figure}

It is well-known that square QAM outperforms other rectangular constellations of the same size (i.e., the product $M_IM_Q$ is kept unchanged) in terms of the BER. Motivated by this, in Fig. \ref{rectangular_QAM}, we examine the impact of phase noise on different rectangular shapes of the same size. In this example, $32\times2$, $16\times4$, and $8\times8$ QAM are taken for demonstration. The result shows that the more balance between signal point sizes on the I and Q channels is, the less performance degradation due to phase noise becomes. For instance, at conditional BER = $10^{-4}$, the performance losses of $32\times2$, $16\times4$, and $8\times8$ QAM are $3.4$ dB, $0.75$ dB, and $0.3$ dB, respectively.  
As such, it is evident that the square QAM also suffers less phase noise-induced performance loss than other rectangular configurations. 
\begin{figure}[ht]
\centering
\includegraphics[scale = 0.315]{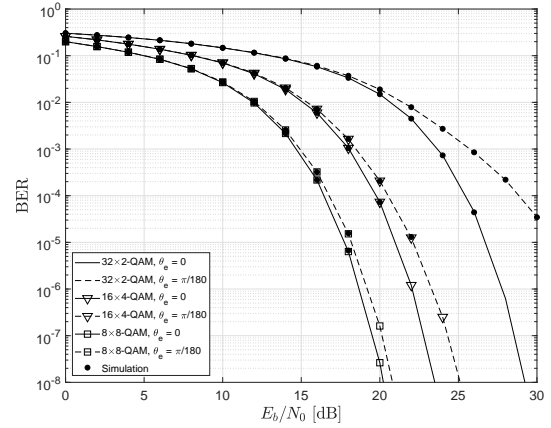}
\caption{BER of different rectangular QAM  configurations. }
\label{rectangular_QAM}
\end{figure}
\section{Conclusion}
We have derived an exact closed-form expression for the conditional BER of rectangular QAM given a value of phase noise. The derivation was based on analytical characterizations of bit decision regions and positions. Numerical results showed that significant performance losses are possible in the case of large constellation sizes. Besides, with the same constellation size, the square QAM is preferable in terms of both performance and loss due to phase noise. 
\bibliographystyle{ieeetr}
\bibliography{ref}
\end{document}